\documentclass[11pt]{article}
\usepackage{algorithm,algorithmic}

\newcommand{\ceil}[1]{\lceil #1 \rceil}
\newtheorem{lemma}{Lemma}
\newtheorem{theorem}{Theorem}

\newcommand{\BBox}{\rule{0.1in}{0.1in}}
\newenvironment{proof}{\noindent {\bf Proof:}}{\ \BBox \\*}

\begin{document}

\title{A study of integer sorting on multicores}

\author{Alexandros V. Gerbessiotis\thanks{CS Department, New Jersey Institute of Technology, 
Newark, NJ 07102, USA. Email: alexg@njit.edu}
}
\maketitle
\thispagestyle{empty}


\begin{abstract}
Integer sorting on multicores and GPUs can be realized by a variety 
of approaches that include variants of distribution-based methods 
such as radix-sort, comparison-oriented algorithms such as 
deterministic regular sampling and random sampling parallel sorting, 
and network-based algorithms such as Batcher's bitonic sorting 
algorithm. 

In this work we present an experimental study of integer sorting 
on multicore processors.  We have implemented serial and parallel 
radix-sort for various radixes, deterministic regular oversampling 
and random oversampling parallel sorting, and also some previously
little explored or unexplored variants of bitonic-sort and 
odd-even transposition sort.

The study uses multithreading and multiprocessing parallel 
programming libraries with the C language implementations working 
under Open MPI, MulticoreBSP, and BSPlib utilizing the same
source code.

A secondary objective is to attempt to model the performance of 
these algorithm implementations under the MBSP 
(Multi-memory BSP) model.
We first provide some general high-level observations on the 
performance of these implementations. 
If we can conclude anything is that accurate prediction
of performance by taking into consideration architecture 
dependent features such as the structure and characteristics of 
multiple memory hierarchies is difficult and more often than not 
untenable. To some degree this is affected by the overhead imposed
by the high-level library used in the programming effort.
We can still draw however some reliable conclusions 
and reason about the performance of these implementations
using the MBSP model, thus making MBSP useful and usable.
\end{abstract}


\section{Summary}
\label{summary}

Integer sorting on multicores and GPUs can be realized
by traditional distribution-specific algorithms such as radix-sort 
\cite{BLM91,Garber2008,Langr16,Maus11}, or variants of it 
that use fewer  rounds of the baseline count-sort implementation 
provided additional information about key values is available 
\cite{Cheng11,Zhong12}. 
Other approaches include algorithms that use specialized 
hardware or software features of a particular multicore 
architecture \cite{Bramas17, Cheng11,Inoue11,Langr16}.  
Comparison-based algorithms have also been used with some 
obvious tweaks: use of deterministic regular sampling sorting 
\cite{SH} that utilizes serial radix-sort for local sorting 
\cite{Dehne99, Dehne12, Dehne17} or use other methods for local
sorting \cite{Zaghloul17,BLM91,Cederman08,Cheng11,Inoue11}. 
Network-based algorithms such as Batcher's \cite{Batcher68} bitonic 
sorting \cite{Ionescu97,BLM91,Peters11,Rathi16,Cederman08}
have also been utilized. In particualar, bitonic sorting is a 
low programming overhead algorithm and thus more suitable for
GPU and few-core architectures, is simple to implement, and
quite fast when few keys are to be sorted, even if its 
theoretical performance is suboptimal.

In this work we perform an experimental study of integer sorting 
on multicore processors using multithreading and multiprocessing 
based libraries that facilitate parallel programming.
Our implementations need only recompilation of the same C language
source to work under Open MPI \cite{openMPI},
MulticoreBSP \cite{Yzelman14}, and a multi-processing and out 
of maintenace library, BSPlib \cite{Hill96}. 

Towards this we have implemented  
(a) serial and parallel radix-sort for various radixes, named 
SR4 for the serial   radix-256 version,
PR4 for the parallel radix-256,  and 
PR2 for the parallel radix-65536 versions,
(b) previously little explored or unexplored  variants of
bitonic-sort named BTN, and odd-even transposition sort named OET
respectively,
(c) a variant of the  random oversampling parallel sorting 
algorithm of \cite{GV94} named GVR,
(d) the deterministic regular
oversampling parallel sorting algorithm of \cite{GS96,GS99a}
named GSD,
and
(e) a random oversampling parallel sorting algorithm named GER
that differs from other random sample sorting approaches in 
that it follows instead the skeleton of deterministic regular
sampling or oversampling methods \cite{SH,GS96,GS99a}.

We then present observations on a performance study of such 
algorithm implementations on multiple platforms and architectures 
and multiple programming libraries. If we can conclude anything 
is that precisely modeling their performance by taking into 
consideration architecture dependent features such as the 
structure and characteristics of multiple memory hierarchies is
difficult and more often than not unusable or unreliable. 
This is primarily because of the uknown or difficult to model
characteristics of the underlying software library that facilitates 
parallel programming.
However  we can still draw some very simple conclusions  using 
traditional architecture independent parallel modeling under 
Valiant's BSP model \cite{LGV90} or the augmented 
MBSP model \cite{G15} that has recently been proposed by this 
author as a flat multi-memory extension of the BSP model for 
multicore or multimemory architectures. 
We have stayed away for example from
Valiant's Multi-BSP model \cite{LGV11} for a variety of
obvious reasons. Modern cache and memory hierarchies change from
one processor generation to another that makes precise modeling
of them a hopeless task. Writing portable code that is optimal
under such conditions is next to impossible.
The Multi-BSP hierarchical modeling
is suitable for hierarchical algorithms. Implementing portable
and efficient hierarchical algorithms is beyond the capabilities 
of this author, a hard task left to others, and frankly  difficult
if not impossible to achieve. Moreover, none of the
algorithm implementations utilized in this study exhibits features 
of a hierarchical algorithm. 
In fact the GVR and GSD variants of \cite{GV94} and
\cite{GS96,GS99a} respectively, purposefully eliminate hierarchical 
features of the original algorithms for the benefit of generating 
a portable and efficient and potentially performance predictable
implementation. The hierarchical features inserted into 
\cite{GV94,GS96,GS99a} served only one purpose: to show
theoretical optimality for extreme conditions when $p$, 
the number of processors is very close to $n$ the number of 
keys to be sorted.
It is for these reasons that we believe that the simplicity of 
MBSP is more suitable for modeling than Multi-BSP.

Some of the experimental conclusions drawn from this study
are discussed briefly below.
For small problem sizes (say $n$, the number of integer keys 
is lower than about 150,000 or so), BTN, the variant 
of bitonic sorting, indeed outperforms serial or parallel
radix-sort with their more time-consuming setup, and
complex communication/memory patterns
for the core and thread count of the testbed
platforms (up to 32 cores or hardware-supported threads). 
This has been observed indepependently by others before \cite{BLM91}.
What has been quite even more extraordinary is 
that the variant OET of odd-even transposition sort implemented
also exhibited good, yet slightly worse performance compared to
BTN on some architectures.
This did not use to be the case in the past when one had to
deal with the  $p$ processors of a cluster, an SMP machine, or 
a supercomputer. It rather seems to be a featur of multicore
architectures.

For small problem sizes, at around 16 cores or threads, 
GSD started consistently beating all variants of radix-sort
and it was either better or had comparable performance to BTN.
GER had a slightly worse performance than GSD yet it was
also consistently better than all variants of radix-sort as well.
GVR exhibited marginally worse performance than GER or GSD.
This however might have to do with the different sampling 
methods and sample sizes used in GVR and GER and also in GSD. 
We shall study such behavior in another work.

For large problem sizes (say $n$ is at least 8,000,000 or more),
BTN and OET were in general uncompetitive. 
PR4 had an edge over the 
sampling based algorithms for large and increasing problem sizes,
32,000,000  to 128,000,000 but marginally only in most cases.
PR2 had performance worse than BTN and OET in several cases.
This is primarily a result of the high radix of PR2; 
more optimizations are needed for PR2 to become competitive.

Moreover, we have observed that assigning multiple threads per 
core is not recommended for CPUs with large number of cores. 
This is because of  the additional demand to access a
scarce resource, main memory (RAM), in data intensive 
applications that might also not be locality friendly.
For CPUs with moderate number of cores, 
assigning multiple threads per core
should never exceed the hardware supported bound 
(usually two), or be even lower than that.
If the number of cores is  small, multiple threads
per core can be used as long as problem size is kept small,
i.e. the application becomes less data intensive.

For parallel radix-sort of 32-bit integers,
radix-$2^8$ radix-sorting i.e. four rounds of baseline 
count-sort is faster than the alternative radix-$2^{16}$ 
sorting that uses two rounds. 
However, depending on the architecture and its structure 
of its caches (level-1, level-2 and level-3) it is 
possible that radix-$2^{16}$ to outperform radix-$2^8$.
This has not been observed in our experiments though.

Overall efficiency is dependent on the number of cores if
the degree of parallelism is large. For degree 
of parallelism less than four or eight, efficiency 
can be expressed either in terms of number of cores or 
threads. As most architectures only support in hardware
two threads per core, inefficiencies and significant
drop in performance arise if one goes beyond that.
For the multiprocessing library, BSPlib, exceeding the number
of hardware cores leads to immediate drop in performance.
In this latter case number of cores rather than number
of threads determines or best describes speedup or 
efficiency.

\section{Related work}
\label{relatedwork}

The problem of sorting keys in parallel has been studied
extensively. One major requirement in parallel algorithm
design is the minimization of interprocessor communication
that reduces non-computational overhead and thus speeds up
parallel running time. In order to shorten communication
several techniques are employed that allow coarse
(eg. several consecutive keys  at a time) rather than fine-grained
(eg. few individual keys) communication.
Coarse-grained communication usually
takes advantage of locality of reference as well.
Some of these optimizations work transparently in multicore
or manycore programming as well.
In parallel sorting, if one wants to sort $n$ keys
in parallel using $p$ processors, one obvious way to achieve
this is to somehow split the $n$ keys into $p$ sequences of
approximately the same size   and then sort these $p$ sequences
independently and in parallel using either the same algorithm
recursively in a hierarchical fashion but in a serial setting,
 or resort to another fast serial algorithm, so that the 
concatenation (rather than merging) of the $p$
sorted sequences generates the sorted output.
For $p=2$ this is  quicksort \cite{Hoare62}.

The experimental work of \cite{BLM91} offers a collection 
of parallel algorithms that have been used unmodified 
or not as a basis for integer multicore or parallel sorting. 
One such algorithm has been radix-sort, another one has been
sample-sort \cite{Reischuk85,HC,Reif87} i.e. random
oversampling sorting along the lines of \cite{Reif87}.
Key sampling employed in parallel sorting has been studied
in \cite{HC,Reif87,Reischuk85} that provide algorithms with
satisfactory and scalable theoretical performance.
Using the technique of {\em random oversampling}, 
fully developed and refined in the context of parallel 
sorting \cite{HC,Reischuk85,Reif87},
in order to sort $n$ keys with $p$ processors,  one uses a 
sample of $ps-1$ uniformly at random selected keys where $s$ 
is the {\em random oversampling factor}. After sorting the 
sample of size $ps-1$, one then identifies
$p-1$ splitters as equidistant keys in the sorted sample.  
Those $p-1$ keys split the input into $p$ sequences
of approximately the same size that can then be sorted
independently. 
In  \cite{Reif87} it is  shown that the $p$ sequences of the
$n$ input keys induced by the $p-1$ splitters 
will retain with high probability $O(n/p)$ keys each
and will thus be balanced in size up to multiplicative constant
factors if one regulates $s$ properly.
In \cite{GV94}, it is  shown that by fine-tuning and
carefully regulating this oversampling factor $s$,
the $p$ sequences have with high probability
$(1+ \epsilon ) n/p$ keys resulting in a finer balance.
Parameter  $0< \epsilon <1$ is a parameter that is controlled
by $s$.
The bounds on processor imbalance during sorting for 
\cite{GV94} are tighter than those of any other random sampling 
or oversampling algorithm \cite{HC,Frazer70,Reif87,Reischuk85}.

Thus if one was to extend say a traditional quick-sort method 
to use oversampling, one would need to add more steps or 
phases to the sorting operation.
These would include a sample selection phase, 
a sample-sort phase, a splitter selection phase, 
followed by a phase that splits the input keys around the 
chosen splitters.
A subsequent routing operation of the split keys is followed
by local sorting if possible or using the same algorithm
recursively.

{\em Algorithm GVR.} The random oversampling algorithm in 
\cite{GV94} follows this execution pattern.
\cite{GV94} generically indicates how one can sort the 
sample using algorithms from \cite{NS}. If $p^2 < n$, and 
it is so for the case of our  multicore sorting study, one 
could perform sample sorting either serially or in parallel 
thus bypassing \cite{NS} methods.
In our GVR adaptation of \cite{GV94} we thus use bitonic-sort
i.e. BTN for sample sorting.  If one performs a binary search
of the input keys to the $p-1$ splitters selected  in random
sampling the resulting communication would be rather fine-grained.
There is a need to pack together all keys destined to the
same processor. Thus in our implementation, following the
binary search step we ''count-sort'' the keys based
on the destination processor each one is destined to.
This alas requires $n/p$ extra space for the processor identifiers
and $n/p$ more space for the output of ''count-sorting''.
At that point keys are ordered by destination into 
$p$ sequences and then sequence $i$ is dispatched to 
processor $i$. If this is not realized,
fine-grain communication will degrade the performance of GVR.
Furthermore, in the context of
integer sorting the split keys are locally sorted using 
radix-sort i.e. SR4; there is no need for the recursion
of \cite{GV94}.

{\em Deterministic regular sampling} as presented in \cite{SH} 
works as follows.  First split regularly and evenly $n$ input keys 
into $p$ sequences of equal size  (the first $n/p$ or so keys  become
part of the first sequence, the next $n/p$ keys
part of the second and so on).  Locally sort the $p$  sequences 
independently (eg by using radix-sort if the keys are integer), 
and then pick from each sequence $p-1$ regular and equidistant 
sample keys for a sample of total size $p(p-1)$.
A sorting (or multi-way merging) of these sample keys can be
followed by the selection of $p-1$ equidistant splitters
from the sorted sample.
Those splitters are broadcast to all processors or cores,
that then split their input keys around those $p-1$ splitters into
$p$ buckets. The keys of bucket $i$ are routed to processor or core $i$
in a follow-up communication step. The $p$ received sequences by 
processor $i$, one per originating processor, are then merged using 
$p$-way merging and the algorithm terminates.
One can then prove that if the $n$ keys are
split  around  the $p-1$ splitters,
then none of the processors including processor $i$ will receive more
than $2n/p$ keys \cite{SH}.  For this method to be optimal
one needs to maintain $n/p > p^2$. The work of \cite{Dehne99} also
discusses the case of $n/p < p^2$.

{\em Deterministic regular oversampling.}
In \cite{GS96,GS99a,GS97a} the notion of regular sampling is
extended to include oversampling thus giving rise to
deterministic regular oversampling.
In that context, $(p-1)s$ regular and equidistant sample keys are
chosen from each one of the $p$ sequences thus generating a larger 
sample of  $p(p-1)s$ keys deterministically,
with $s$ being now the {\em regular oversampling factor}. One can
claim \cite{GS99a}  that each one of the $p$ sequences split 
by the $p-1$ resulting splitters is of size   no more than
$(1+\delta ) n/p$, where $\delta >0$ also depends on the choice of 
$s$ thus reducing a $2n/p$ imbalance into an $(1+\delta) n/p$ one
with $\delta$ depending on $s$.

{\em Algorithm GSD.} In this study the algorithm of \cite{GS99a}
is implemented quite faithfully including the  parallel
sample sorting step that involves the BTN implementation.
The initial local sorting uses SR4 \cite{CLRS}. More details follow in
the next section.

{\em Algorithm GER.}
Random oversampling-based sorting is supposed to be superior to
deterministic regular sampling or oversampling-based sorting as
the regular oversampling factor $s$ can not be finely-tuned as much
as the  random oversampling factor $s$ \cite{GS96,GS99a}.
Moreover the deterministic algorithms suffer from the final multi-way
merging step.
To test  this, in this study we also implement a second random
oversampling parallel sorting algorithm, GER, that works
however quite differently from GVR and other random
sampling algorithms and its skeleton draws from
deterministic regular sampling \cite{SH} or 
regular oversampling \cite{GS96,GS99a} i.e. it is very similar
to GSD.
In such an approach local key sorting comes first before
sample and splitter selection,
followed at the very end by a multi-way merging that is
less locality-sensitive. Thus the only difference between GER
and GSD is the sample selection step: the former's one resembles
that of GVR's. For both cases however BTN 
is used for sample sorting. More details follow in the next section.

The algorithm of \cite{SH} was used by \cite{Dehne99} in the
context of integer sorting; initial local sorting involved
radix-sort.  Subsequently \cite{Dehne12,Dehne17} utilized this
approach for GPU sorting of arbitrary (not necessarily integer) keys.
Similarly to the approach of \cite{GS96,GS99a} but differently,
a sample of size $ps$ is being used rather than
the $p(p-1)s$ of \cite{GS96,GS99a}. The GPU architecture's
block thread size determines $p$ and other GPU constraints
dictate $s$. Thus effectively the implied oversampling
factor of \cite{GS96,GS99a} becomes
$s/(p-1)$ in \cite{Dehne12,Dehne17}. \cite{Peters11} also performs
GPU sorting by using a bitonic sorting method. The conclusions
of the two papers  \cite{Peters11,Dehne12}
agree that bitonic-sorting works better for small
values of $n$ and sample sort \cite{Dehne12,Dehne17}
or radix-sort \cite{Peters11} is better for larger $n$.
Their overall conclusions are in line with
those of \cite{BLM91}.

In \cite{Zaghloul17} a parallel merge-sort is analyzed
and implemented on multicores. The parallelization
of merge-sort is not optimal. An $(n/p)\lg{(n/p)}$
local and independent sorting on $p$ threads/cores
is followed by a merging yielding
a speedup bounded by $n\lg{(n)} / (n \lg{(n)}/p +2n)$.
This work is similar to that of \cite{Zhong12}.
The latter deals with multisets (keys taking a range
of distinct values).  Even on four processing
cores efficiency is less than $50\%$.
\cite{Bramas17} utilizes a quicksort approach and
AVX-512 instructions on Intel's Knights Landing.
For small problem sizes the use of insertion sort
is replaced by bitonic sort to allow for
vectorization.
This is close to the work of \cite{Cederman08} that introduces
GPU-Quicksort. The thread
processors perform a single task: determining whether
a key is smaller or not than a splitter. Then
a rearranging of the keys is issued following something akin
to a scan/parallel prefix \cite{Leighton91} operation. Its
performance is compared to  a radix-sort implementation and
shown to be slightly better but mostly worse than
a Hybridsort approach that uses bucketsort and merge-sort
(whose performance depends on the distribution of the
input keys).

A third algorithm used in the pioneering work of \cite{BLM91} 
is bitonic sorting.
For sorting $n$ keys (integer or otherwise) the bitonic
sort of \cite{BLM91} employs a $\Theta(\lg^2{(n)})$ stage
bitonic sort. If the input is regularly partitioned,
$n/p$ keys are located inside a single
processor. Thus bitonic merging of those keys can be done
entirely within the corresponding processor's memory as 
observed in \cite{BLM91}; the authors alternatively proposed
using linear-time serial merging over the slower 
bitonic merging.
It is worth noting that the bitonic sort of \cite{BLM91}
for small $n$ outperformed the other sorting
methods as implemented on a Connection Machine CM-2.
This is due to the low programming overhead of the method as
also highlighted in the previous section.
This approach of \cite{BLM91} to bitonic sorting has 
been followed since then. More recently  
\cite{Rathi16} is using an implementation drawn 
from \cite{Ionescu97} for 
bitonic sorting on GPUs and CPUs.
Such an implementation resembles the one above where 
for $n/p$ keys local in-core operations are involved.
Alas, overall GPU performance is rather unimpressive: 
a 10x speedup over the CPU implementation
on an NVIDIA GT520, or 17x on a Tesla K40C for 5M
keys, and a modest 2-3x speedup for fewer keys.

However neither \cite{BLM91} nor the other 
implementations of bitonic sort cited
\cite{Rathi16,Ionescu97} or to be cited
later in this section considered that bitonic sorting
involving  $n$ keys on $p$ 
processors/ cores/ threads
can utilize a bitonic network of 
$\Theta(\lg^2{(p)})$ stages.
If $p$ is substantially smaller than $n$, then
the savings are obvious compared to a  
$\Theta(\lg^2{(n)})$ stage bitonic sorter utilized in 
\cite{BLM91,Rathi16,Ionescu97} and other works.

{\em Algorithm BTN.}
Indeed this author \cite{GS96,GS99a} highlighted the
possibility and used  $\Theta(\lg^2{(p)})$ stage bitonic 
sorting for sample sorting involving more than $p$
keys in the context of bulk-synchronous parallel 
\cite{LGV90} sorting.
In particular, \cite{GS96,GS99a} cite the work of 
\cite{BS78,Knuth73} for first observing that
a $p$-processor bitonic sorter can sort $n$ keys in
$\lg{(p)} (\lg{(p)}+1)/2$ stages (or rounds). 
Think of sorting with this $p$-processor bitonic
sorter $p$ 'keys' where each 'key' is a sorted sequence 
of $n/p$ regular keys.  
A comparison of two 'keys' on two processors gets resolved 
by merging the two $n/p$ sorted sequences located  to
these two processors and then assigning the lower
half, the smaller keys, to the lower indexed 
(eg.  left processor) and the upper half to the 
higher indexed processor (eg.  right processor).
At start-up before bitonic sorting commences 
each one of the $p$ 'keys' needs to be in the right
order i.e. the $n/p$ keys of each processor are 
locally sorted.  This local sorting can utilized
merge-sort or if the keys are integer radix-sort, or
something else.
In our realization of this $\lg{(p)} (\lg{(p)}+1)/2$
stage bitonic sorter, local initial sorting is
realized by SR4.
We have been calling this implementation BTN.

{\em Algorithm OET.}
Odd-even transposition sort \cite{Leighton91} is an
unrefined version of bubble-sort that has
been used for sorting in array structured parallel 
architectures
(one-dimensional arrays, two-dimensional meshes, etc).
In such a sort $n$ keys can be sorted by $n$ processors
in $n$ rounds using an oblivious sorting algorithm.
In an odd-indexed round a key at an odd-indexed position
is compared  to the even-indexed key to its immediate
right (index one more); a swap is performed if the former 
key is greater.  Like-wise for an even-indexed round. 
A very simple observation we make is that if the number 
of processors/cores is $p$, then a $p$-wide $p$-round  
odd-even transposition sort can sort $n$ keys by
dealing with $n/p$ key sequences, one such 
sequence per processor as it was done in the case of BTN
above. Each one of the $p$ cores is assigned one 'key', 
a sorted sequence of $n/p$ regular keys. Comparisons
between 'keys'  get resolved  the same way as in BTN
by merging the two sorted sequences of the two 'keys'
involving $n/p$ regular keys each.
Then the smaller $n/p$ keys get separated from the  
larger $n/p$ keys; the former are destined to the
lower-indexed, and the latter to the  
next-indexed processor respectively.
Likewise as before, at start-up the $p$ 'keys'
must be sorted before the $p$-wide odd-even transposition 
sort commences its execution on the $n$ key input. 
This is a local sort operation.
We call this $p$-wide implementation of odd-even
transposition sort used to sort $n$ keys OET for 
future references; local initial sorting utilizes SR4 \cite{CLRS}.

The implementation of AA-sort is undertaken in 
\cite{Inoue11}. AA-sort can be thought of as 
a bubble-sort like enhancement of the odd-even 
transposition sort we discussed earlier
and called OET. Whereas in OET we have
odd and even phases in AA-sort there is no such
distinction and either a bubble-sort step 
(from left to right) is executed if a gap parameter 
$g$ has a value of 1, or non-adjacent keys are 
bubble-sorted if $g$ is greater than 1. 
(Thus $a[i]$ and $a[i+g]$ are then compared.)
Likewise to OET initially each $n/p$ sequence is 
sorted; in the case of AA-sort, a merge-sort is
used. 
The use of a bubble-sort oriented approach is to
exploit vectorization instructions of the 
specific target platforms: 
PowerPC 970MP and Cell Broadband Engine(BE).
In the implementations \cite{Inoue11} present also 
results for a bitonic sort based implementation.
AA-sort seems to be slightly faster than bitonic sort 
with SIMD enhancements for 16K random integers. 
In the Cell BE, AA-sort outperforms bitonic sort 
for 32M integers (12.2 speedup for the former 
over 7.1 for the latter on 16 CELL BE cores).

The AQsort algorithm of \cite{Langr16} utilizes 
quicksort, a comparison-based sorting algorithm,
and  OpenMP is utilized to provide a 
parallel/multithread version of quicksort. 
Though its discussion in an otherwise integer
sorting oriented works as this one might not make
sense, there are some interesting remarks
made in \cite{Langr16} that are applicable to
this work.  It is observed
that hyperthreading provides no benefit and that
for Intel and AMD CPUs best performance is obtained
for assigning one thread per core. For Intel Phi
and IBM BG/Q two threads per core provide marginally
lower running times even if four threads are hardware
supported.

In \cite{Maus11} radix-sort is discussed in the
context of reducing the number of rounds of
count-sort inside radix-sort by inspecting
key values' most significant bits. A
parallelized radix-sort along those lines
achieves efficiencies of approximately
15-30\% (speedup of 5-10 on 32 cores).
Other conclusions are in line with 
\cite{Langr16} in that memory channels can't
keep up with the work assigned from many
parallel threads.

\section{Implementations}
\label{implementations}

\subsection{MBSP model}

In this section we shall describe the algorithms that
we shall implement and analyze their performance under the
{\em Multi-memory Bulk-Synchronous Parallel (MBSP) } model 
of computation \cite{G15}.
The MBSP is  parameterized by the septuplet
$(p,l,g,m,L,G,M)$ to abstract, computation and memory interactions
among multi-cores.
In addition to the modeling offered by the BSP model \cite{LGV90} 
and abstracted
by the triplet $(p,l,g)$, the collection of core/processor   components
has $m$ alternative memory units distributed in an arbitrary way,
and the size of the ``fast memory'' is $M$ words of information.
The cost of memory unit-related I/O is modeled by the pair $(L,G)$.
$L$ and $G$ are similar to the BSP parameters $l$ and $g$ respectively.
Parameter $G$ expresses the unit transfer time
per word of information and thus reflects the memory-unit
throughput cost
of writing a word of information into a memory unit
(i.e. number of local computational operations as units of time
per word read/written).
Parameter $L$ abstracts the memory-unit access latency time
that reflects access delays contributed mainly but not exclusively by
two factors:
 (a) unit access-related delays that can not be hidden or amortized by
     $G$, 
and
 (b) possible communication-related costs involved in
accessing  a non-local unit as this could require
intraprocessor or interprocessor communication.

Using the MBSP cost modeling generic cache performance will be
abstracted by the pair $(L,G)$. 
Parameter $m$ would be set to $p$
and $M$ will be ignored; we will assume that $M$ is large enough
to accommodate the radix-related information of radix-sort.
Intercore communication will be abstracted by $(p,l,g)$. Since
such communication is done through main memory, $g$ would be the
cost of accessing such memory (aka RAM). 
We shall also ignore $l$ and $L$.
This is possible because in integer-sorting the operations 
performed are primitive and interaction with memory is 
the dominant operation.
Thus the cost model of an algorithm would abstract only cost of
access to the fast memory ($G$) and cost of access to the slow memory
($g$). Then we will use the easy to handle $g=5G$ to further
simplify our derivations. This is  based on the
rather primitive thinking that  $20ns$ and $100ns$ reflect access times
to a cache (L2 or higher) and main memory respectively thus defining
a ratio of five between them.

\subsection{SR4, PR2, PR4, BTN and OET under MBSP}

\noindent
{\bf Serial radix-$r$ radix-sort: SR4}

A serial     radix-sort (previously called SR4) is implemented and used
for local independent sorting in the
odd-even transposition sort and bitonic sort implementations.
The radix used is $r=256$ i.e. it is a four-round count-sort.
In each round of count-sort the input is read twice,
first during the initial counting process and last when the output is
to be generated, and the output is generated once. Thus the cost
of such memory accesses is $3Ng$, with $g$ referring to the cost of
accessing the main memory and accounts for two input and one output operation. 
Moreoever allocation and initialization of
the count array incurs a cost of $2rG$, with $G$ being the cost of accessing
the fast cache memory. We shall ignore this cost that is dominated by other
terms.
During the count operation the count array is
accessed $N$ times and so is during the output operation for a total cost
of $2NG$.
Thus the overall cost of a round is $3Ng+2NG$.
For all four rounds of 32-bit sorting the total cost is given by the following.
\begin{eqnarray*}
\label{Ts}
 T_s (N,g,G,r) &=& \left( 32 / \lg{(r)} \right) \cdot  \left( 3Ng+2NG \right)
\end{eqnarray*}
If $g=5G$ and $r=256$ then
\begin{equation}
\label{TsG}
 T_s (N,G) = 68 N G
\end{equation}


\noindent
{\bf Parallel radix-$r$ radix-sort: PR4 and PR2}

We shall denote with PR2 and PR4 radix
$r=2^{16}$ and $r=2^8$ parallel radix-sort algorithms.
Ignoring some details that are implementation
dependent such as the use of counters in the
serial part and local and remote copies used in the parallel part,
we recognize a cost $2rpg$ due to  scatter and gather operations 
involved in the parallel part the algorithm. 
If $n$ keys are to be sorted, each processor or core is assigned
roughly $N=n/p$ keys. A $2NG$ cost is assigned for the same reasons 
that was assigned in the serial version. 
A $3Ng$ of the serial version will
become $4Ng$ to account for a communication required before the 
output array is formed in a given round of count-sort. 
\begin{eqnarray*}
\label{Tp}
 T_p (N,g,G,p,r) &=& \left( 32 / \lg{(r)} \right) \cdot  \left( 4Ng + 2NG + 2prg \right)
\end{eqnarray*}
If $g=5G$ and $r=256$ then
\begin{equation}
\label{Tp4}
 T_p (n,G,p) = \left( 88n/p + 40\cdot  256 \cdot p \right) G
\end{equation}
If $g=5G$ and $r=256^2$ then
\begin{equation}
\label{Tp2}
 T_p (n,G,p) = \left(  44n/p + 20 \cdot 256^2 \cdot p \right) G
\end{equation}

\noindent
{\bf Odd-even transposition sort: OET}

We analyze the algorithm previously referred to as OET.
If $n$ keys are to be sorted, each processor or core is assigned
roughly $N=n/p$ keys. First the $N$ keys per processor or core are
sorted using a radix $r=256$ radix-sort independently and in parallel
of each other that requires time $T_s (n/p,G)$. 
Then a $p$ round odd-even transposition sort takes place utilizing
$n/p$ sorted sequences as explained earlier for OET.
One round of it requires roughly $4Ng$ 
for communication and  merging (two input and one output arrays). 
Thus the overall cost of all $p$
phases of OET will be as follows.
\begin{eqnarray*}
\label{To}
 T_o (n,g,G,p) &=& T_s (n/p,G) + p \left( 4n/p  \right)g 
\end{eqnarray*}
If $g=5G$ and $r=256$ then
\begin{equation}
\label{ToG}
 T_o (n,G,p) = \left( 68 n/p  + 20 n  \right) G
\end{equation}

\noindent
{\bf Bitonic Sort: BTN}

We analyze the algorithm previously referred to as BTN.
If $n$ keys are to be sorted, each processor or core is assigned
roughly $N=n/p$ keys. First the $N$ keys per processor or core are
sorted using a radix $r=256$ radix-sort independently and in parallel
of each other that requires time $T_s (n/p,G)$. 
Then  $\lg{(p)} (\lg{(p)}+1)/2$ stages of a $p$-processor bitonic-sort
are realized as explained in Section~\ref{relatedwork}.
One round of it requires roughly $4Ng$ 
for communication and  merging/comparing 
(two input and one output arrays). 
Thus the overall cost of all 
stages of bitonic sort will be as follows.
\begin{eqnarray*}
\label{Tb}
 T_b (n,g,G,p) &=& T_s (n/p,G) + \left( \lg{(p)} \cdot (\lg{(p)}+1) /2 \right) \cdot \left( 4n/p  \right) g
\end{eqnarray*}
If $g=5G$ and $r=256$ then
\begin{equation}
\label{TbG}
 T_b (n,G,p)  = \left( 68 n/p + \left( 10 n\lg{(p)}(\lg{(p)}+1) \right) / p \right) G
\end{equation}

\subsection{Oversampling algorithms: GSD, GVR, GER}

GSD is depicted in Algorithm~\ref{GSD}. 
Within GSD parameter $s$ is the regular oversampling factor whose value 
is regulated through the choice of $r=\omega_n$.
Parameter $r=\omega_n$ could have been included in the parameter list
of GSD.  The theorem  and proof that follow simplify the proof
and the results shown in a more general context in  \cite{GS96,GS99a}.

%
%
\begin{algorithm}
\caption{{GSD $(X,n,p)$  \{sorts $n$ integer keys of $X$ \} }}
\label{GSD}
\begin{algorithmic}[1]
\STATE {\sc LocalSorting.} The $n$ input keys are regularly and evenly 
split into $p$ sequences each one of size   approximately $n/p$. 
Each sequence is sorted by SR4. Let $X_k$, $0 \leq k \leq p-1$, 
be the $k$-th sequence after sorting.
\STATE {\sc Sample Selection.}  
Let $r =  \lceil \omega_n \rceil $ and $s=rp$. Form locally a sample 
$T_k$ from the sorted $X_k$.  
The sample consists of $rp-1$ evenly spaced keys of $X_k$ that 
partition it into $rp$ evenly sized segments; 
append the maximum of the sorted $X_k$ (i.e. the last key) into 
$T_k$ so that the latter gets $rp$ keys.  Merge, in fact sort, 
all $T_k$ into a sorted sequence $T$ of $rp^2$ keys using BTN.
\STATE {\sc Splitter Selection.}
Form the splitter sequence $S$ that contains the 
$(i\cdot s)$-th smallest keys of $T$, $1 \leq i \leq p-1$, 
where $s=rp$. Broadcast splitters.
\STATE {\sc Split input keys.}
Split the sorted $X_k$ around $S$ into sorted subsequences 
$X_{k,j}$, $0 \leq j \leq p-1$, for all $0 \leq k \leq p-1$ 
using binary search. Processor $k$ sends to processor 
$j$ sequence $X_{k,j}$.
\STATE  {\sc Merging.} Subsequences $X_{k,j}$ for all 
$0 \leq k \leq p-1$ are merged into $Y_j$, for 
all $0 \leq j \leq p-1$ using $p$-way merging. 
The concatenation $Y$ of all $Y_j$ is returned.
\end{algorithmic}
\end{algorithm}
\begin{theorem}
\label{detsorting}
For any $n$ and $p \leq n$, and any function
$\omega_{n}$ of $n$ such that $\omega_{n}=\Omega(1)$,
$\omega_n = O(\lg{n})$ and $p^2 \omega^2_{n} \lg^2{n} = o(n)$,
operation GSD   
requires time at least 
$T_s (n/p,G) + 5Gn_{\mathit{max}} \lg{p}$, 
plus low order terms that are  $o(T_s(n/p,G))$, where
$n_{\mathit{max}} =(1+1/\ceil{\omega_{n}})  (n/p) +\ceil{\omega_n} p$.
\end{theorem}

\begin{proof}{}
The input sequence is split arbitrarily into $p$ sequences of about
the same size   (plus or minus one key). This is step 1
of GSD.
Parameter $r$ determines the desired upper bound in
key imbalance of the $p$ sorted sequences $Y_k$
that will form the output.
The term $1+1 /r = 1+ 1/ \lceil \omega_n \rceil$ that will characterize 
such an imbalance is also referred to as {\em bucket expansion}
in sampling based randomized sorting algorithms \cite{BLM91}.
In the discussion to follow we track constant values for
key sorting and multi-way merging but use asymptotic notation for
other low-order term operations.
In step 1, each one of the $p$ sequences is sorted independently of
each other using SR4.   
As each such sequence is of size   at most  $\lceil n/p \rceil$, this step
requires  time  $T_s (n/p,G)$.

Subsequently, within each sorted subsequence $X_k$,
$\ceil{\omega_{n}} p - 1=rp-1$ evenly spaced sample keys are selected,
that partition the corresponding
sequence into $r p$ evenly sized segments.  Additionally,
the largest key of each sequence is appended to $T_k$.
Let $s = rp$ be the size   of the
so identified sequence $T_k$.
Step 2 requires time $O(s)g$ to perform. 
The $p$ sorted sample sequences, each consisting of $s=rp$ sample keys,
are then merged/sorted using BTN in time $T_b (rp^2, G,p)$ into $T$.
Let sequence $T= \langle t_{1}, t_{2}, \ldots, t_{ps}\rangle$ be
the result of that  operation.
In step 3,  a sequence $S$ of evenly spaced splitters is formed
from the sorted sample by picking as splitters keys
$t_{is}$, $1 \leq i  < p$. This step takes time $O(p)g$.
Step 4 splits $X_k$ around the sample keys in $S$.
Each one of the $p$ sorted sequences  decides the position of every key
it holds with respect to the $p-1$ splitters by
way of serial     merging the $p-1$ splitters of $S$ with the input
keys of $X_k$ in $p-1+n/p$ time per sequence. Alternately this
can be achieved by
performing a binary search of the splitters into the sorted keys
in time $p \lg{(n/p)}$, and subsequently counting the number of keys
that fall into each one of the $p$ so identified subsequences induced
by the $p-1$ splitters.
The overall running time of this step 
is $p \lg{(n/p)}g$ if binary search is performed.

In step 4,  $X_{k,j}$ is the $j$-th sorted subsequence of
$X_k$ induced by $S$. This subsequence will become part of the
$Y_j$-th output sequence in step 5.
In  step 5, $p$ output sequences $Y_j$ are formed that will eventually
be concatenated. Each such output sequence $Y_j$ is formed from the
at most $p$ sorted subsequences $X_{k,j}$ for all $k$, formed in step 4.
When this step is executed, by way of Lemma~\ref{balance} to be shown
next, each $Y_j$ will comprise of at most
$p=\min{\{p, n_{\mathit{max}}\}}$ sorted subsequences  $X_{k,j}$ for a
total of at most $n_{\mathit{max}}$ keys for $Y_j$, and $n$ keys for $Y$,
where
$n_{\mathit{max}} =(1+1/\ceil{\omega_{n}})  (n/p) +\ceil{\omega_n} p$.
The cost of this step is that of multi-way merging i.e.
$n_{\mathit{max}} \lg{p} g$.

{\sc LocalSorting} and {\sc Merging} thus contribute
$T_s (n/p,G) + n_{\mathit{max}} \lg{p} g$, noting $g=5G$.
Sample selection and sample-sorting contributions amount to $T_b (rp^2, G,p)$.
Other contributions such as $O(pG)$ of step 3 and $5p \lg{(n/p)}G$
can be ignored.
\end{proof}

It remains to show  that at the completion of  step 4 the input keys are
partitioned into (almost) evenly sized subsequences.
The main result is summarized in the following lemma.

\begin{lemma}
\label{balance}
The maximum number of keys $n_{\mathit{max}}$ per output sequence  $Y_j$
in GSD is given by $(1+1/\ceil{\omega_{n}})  (n/p) +\ceil{\omega_n} p$, 
for any $\omega_{n}$ such that
$\omega_{n} = \Omega(1)$ and $\omega_n = O(\lg{n})$, provided that
$\omega_n^2 p = O(n/p)$ is also satisfied.
\end{lemma}

\begin{proof}{}
Although it is not explicitly mentioned in the description of
algorithm GSD we may assume that we initially pad the
input so that each sequence has  exactly $\ceil{n/p}$ keys. At most
one key is added to each sequence (the maximum key can be such a
choice).  Before performing the sample selection operation, we also pad
the input so that afterwards, all segments have the same number of keys
that is, $x =\ceil{\ceil{n/p}/s}$. The padding operation requires time
at most $O(s)$, which is within the lower order terms
of the analysis of Theorem~\ref{detsorting}, and therefore, does not
affect the asymptotic complexity of the algorithm.  We note that padding
operations introduce duplicate keys; a discussion of duplicate handling
follows this proof.

Consider an arbitrary splitter $t_{is}$, where $1 \leq i < p$.
There are at least $i s x$ keys which are not larger than $s_{is}$,
since there are $i s$ segments
each of size   $x$ whose keys are not larger than $s_{is}$.  Likewise,
there are at least $(ps - i s - p + 1) x $ keys which are not smaller
than $s_{is}$, since there are $p s - i s - p + 1$ segments each
of size   $x$ whose keys are not smaller than $s_{is}$.  Thus, by noting
that the total number of keys has been increased (by way of padding
operations) from $n$ to $p s x$, the number of keys
$b_{i}$ that are smaller than $s_{is}$ is bounded as follows.
\[
 i s x \leq b_{i} \leq p s x - \left( ps- i s - p + 1 \right)   x.
\]
A similar bound can be obtained for $b_{i+1}$.  Substituting
$s= \ceil{\omega_{n}} p$ we therefore
conclude the following.
\[
  b_{i+1} - b_{i  } \leq  sx +px-x \leq sx+px =
    \lceil\omega_{n}\rceil p x + p x.
\]
The difference $n_{i} = b_{i+1} - b_{ i }$ is independent
of $i$ and gives the maximum number of keys per split sequence.
Considering that $x \leq (n + p s)/(ps)$ and substituting $s=
\ceil{\omega_{n}} p $, the following bound is
derived.
\[
  n_{\mathit{max}} =
  \left( 1+\frac{1}{\lceil{\omega_{n}}\rceil }\right)
      \frac{n+ps}{p}.
\]
\noindent
By substituting in the  numerator of the previous expression
$s= \ceil{\omega_{n}} p$,
we conclude that the maximum number of keys $n_{\mathit{max}}$ per
output sequence of GSD is bounded above as follows.
\[
  n_{\mathit{max}} = \left( 1 + \frac{1}{\ceil{\omega_{n}}}\right)
    \frac{n}{p} + \ceil{\omega_{n}} p .
\]
\noindent
The lemma follows.
\end{proof}


The skeleton of GSD is being used in developing GER.
Random oversampling-based algorithms in the traditional
approach of \cite{HC,Reischuk85,Reif87,GV94} do not
involve a LocalSorting first step that distinguishes
GER from other random oversampling approaches.
For reference, we also outline GVR next.

%
%
\begin{algorithm}
\caption{{GER $(X,n,p)$  \{sorts $n$ integer keys of $X$ \} }}
\label{GER}
\begin{algorithmic}[1]
\STATE {\sc LocalSorting.} The $n$ input keys are regularly and evenly split into $p$
sequences each one of size   approximately $n/p$. Each sequence
is sorted by SR4. Let $X_k$, $0 \leq k \leq p-1$, be the $k$-th sequence after
sorting.
\STATE {\sc Sample Selection.}  Let $s =   2\omega_n^2 \lg{n} $.
Form a sample $T$ from the corresponding $X_k$.  The sample consists of $sp-1$ keys
selected uniformly at random from the keys of all $X_k$.
Sort $T$ using BTN.
\STATE {\sc Splitter Selection.}
Form the splitter sequence $S$ that contains the $(i\cdot s)$-th smallest keys of $T$,
$1 \leq i \leq p-1$.
\STATE {\sc Split input keys.}
Split the sorted $X_k$ around $S$ into sorted subsequences $X_{k,j}$, $0 \leq j \leq p-1$, for all
$0 \leq k \leq p-1$. Processor $k$ sends to processor $j$ sequence $X_{k,j}$.
\STATE  {\sc Merging.} Subsequences $X_{k,j}$ for all $0 \leq k \leq p-1$ are merged
into $Y_j$, for all $0 \leq j \leq p-1$. The concatenation $Y$ of $Y_j$ is returned.
\end{algorithmic}
\end{algorithm}
%
%
%

%
%
\begin{algorithm}
\caption{{GVR $(X,n,p)$  \{sorts $n$ integer keys of $X$ \} }}
\label{GVR}
\begin{algorithmic}[1]
\STATE {\sc Sample Selection.}  
The $n$ input keys are regularly and evenly split into $p$
sequences each one of size   approximately $n/p$. 
Let $X_k$, $0 \leq k \leq p-1$, be the $k$-th sequence.
Let $s =   2\omega_n^2 \lg{n} $.
Form a sample $T$ from the corresponding $X_k$.  The sample consists of $sp-1$ keys
selected uniformly at random from the keys of all $X_k$.
Sort $T$ using BTN.
\STATE {\sc Splitter Selection.}
Form the splitter sequence $S$ that contains the $(i\cdot s)$-th smallest keys of $T$,
$1 \leq i \leq p-1$.
\STATE {\sc Split input keys.}
Split $X_k$ around $S$ into unordered subsequences $X_{k,j}$, 
$0 \leq j \leq p-1$, for all $0 \leq k \leq p-1$. This requires a binary 
search of $X_k$ into $S$ followed by a count-sort oriented approach that 
sorts the keys of $X_k$ with respect to the processor index $j$ they 
are destined to thus determining the corresponding $X_{k,j}$ subsequence. 
Then and only then does the key moves to $X_{k,j}$.
Processor $k$ sends to processor $j$ sequence $X_{k,j}$.
\STATE  {\sc LocalSorting.} Subsequences $X_{k,j}$ for all $0 \leq k \leq p-1$ are 
concatenated into $Y_j$, for all $0 \leq j \leq p-1$. Then $Y_j$ is sorted locally
on processor $j$ using SR4.  The concatenation $Y$ of the sorted $Y_j$ is returned.
\end{algorithmic}
\end{algorithm}

Although  partitioning and oversampling in the context of sorting
are well established techniques  \cite{HC,Reischuk85,Reif87},
the analysis in \cite{GV94} summarized in
Lemma \ref{Sampling1} below allows one to
quantify precisely the key imbalance of the output sequences $Y_j$.
Let $X = \langle x_{1}, x_{2}, \ldots , x_{n} \rangle$ be an ordered
sequence of keys indexed such that $x_{i} < x_{i+1}$, for all
$1 \leq i  \leq n-1$. The implicit assumption is that keys are unique.
Let
$Y = \{y_{1}, y_{2}, \ldots, y_{ps-1}\}$ be a randomly chosen
subset of $ps-1 \leq n$ keys of $X$ also indexed such that
$y_{i} < y_{i+1}$, for all $1 \leq i  \leq ps-2$,  for some
positive integers $p$ and $s$. Having
randomly selected set $Y$, a partitioning of  $X - Y$ into $p$
subsets, $X_{0}, X_{1}, \ldots, X_{p-1}$ takes place.
The following result shown in \cite{GV94}
is independent of the distribution of the input keys.

\begin{lemma}
\label{Sampling1}
Let $p \geq 2$, $s \geq 1$, $ps < n/2$, $n \geq 1$,  $0 < \varepsilon < 1$,
$\rho >0$, and

\[
  s \geq \frac{1+\varepsilon}{\varepsilon^{2}} \left( 2 \rho \log{n} +
         \log{(2 \pi p^{2}(ps-1)e^{1/(3(ps-1))})} \right).
\]
\noindent
Then the probability that any one of the $X_{i}$, for all $i$,
$0 \leq i \leq p-1$, is of size more than
$\ceil{(1+\varepsilon)(n-p+1)/p}$ is at most $n^{-\rho}$.
\end{lemma}

To conclude, the analysis of Theorem~\ref{detsorting} is applicable to GER.
For GVR, the choice of $s$ in step 1 of Algorithm~\ref{GVR} guarantees
with high probability that no processor receives in step 4 more than
$n_{\mathit{max}} =(1+1/{\omega_{n}})  (n/p) $ keys. Thus whereas local
sorting in GSD and GER involved the same number of keys per processor
(plus or minus one) and multi-way merging in step 5 was unbalanced,
in the case of GVR binary search over the splitters is quite balanced
and the last step 4 involving local sorting is unbalanced. Therefore we
have the following.
\begin{theorem}
\label{rndsorting}
For any $n$ and $p \leq n$, and any function
$\omega_{n}$ of $n$ such that $\omega_{n}=\Omega(1)$
and $p^2 \omega^2_{n} \lg^2{n} = o(n)$,
operation GER   
requires time at least 
$T_s (n_{\mathit{max}},G) + 5G(n/p) \lg{p}$, 
plus low order terms that are  $o(T_s(n/p,G))$, where now
$n_{\mathit{max}} =(1+1/{\omega_{n}})  (n/p) $.
\end{theorem}

\section{Experiments}

All implementations undertaken in this work are in ANSI C; 
the same code need only be recompiled but does not need to 
be rewritten to work with three 
parallel, multiprocessing or multithreaded programming
libraries: OpenMPI \cite{openMPI}, MulticoreBSP \cite{Yzelman14}, 
and BSPlib \cite{Hill96}.  
The source code is available through the author's web-page \cite{AVG18}.
CONFIG1 was a 2-processor Intel Xeon E5-2660 v2 with 256GiB of memory
providing a total of 20 cores and 40 threads.
CONFIG2 an 8-processor quad-core AMD Opteron 8384 Scientific Linux 7 workstation with
128GiB of memory providing a total of 32 cores/threads.
CONFIG3 was a 2-processor Intel Xeon E5-2630 v4 with 256GiB of memory
providing a total of 20 cores and 40 threads.
CONFIG4 was an  Intel Xeon E3-1245 v6 with 32GiB of memory
providing a total of 4 cores and 8 threads.

The version of OpenMPI available and used is 1.8.4.
The version of OpenMPI used is 1.8.1. 
Version 1.2.0 of MulticoreBSP is used and version 1.4 of BSPlib.
The source code is  compiled using the native gcc compiler
{\tt gcc version 7.3.0} with optimization options
{\tt -O3} {\tt -march=native} and {\tt -ffast-math -funroll-loops}
and using otherwise the default compiler and library installation.

Indicated timing results (wall-clock time in seconds) in the tables 
to follow are the averages of four experiments.
We used modest problem sizes of 
$8192\times 10^3 , 32768 \times 10^3 , 131072 \times 10^3$
integers. We shall refer to them in the remainder as 8M, 32M, and 128M
but we caution what this means for the corresponding values ($1M=1024000$).
This is the total problem size, not the per processor size. 
We have also run some experiments for smaller problem sizes on CONFIG3
ranging from 8192 up to 131072 keys over all cores/threads.
Parameter $p$ indicates the number of threads, cores or processes utilized.
For CONFIG1 - CONFIG3 we used $p=4, 8, 16, 32$ and for 
CONFIG4 $p=2, 4, 8$.
For the serial algorithm  SR4 we report wall clock time.
For the other algorithms i.e. PR4, PR2, BTN, OET, GSD, GVR and GER
we only report speed up figures obtained from wall clock time and the
corresponding $p$.
We first provide a  summary  of the results for the four configurations.
Then we list a set of observations that we have drawn from this
experimental study.

For CONFIG1, a clear winner is PR4 for large $n$ and $p$. The maximum
observed speedup was close to the number of physical cores 
(speedup 16.58 for 20 cores) and when the number of threads was equal to the 
maximum of 2 threads per cores ($p=32$ threads). The absolute
maximum speedup was observed for BSPlib closely followed by MulticoreBSP
and openMPI. 
It is quite surprising that an out of maintenance library that uses
multiprocessing can still perform at high level.
The latter had worse results for 8M sizes relative to the other
two libraries. PR2 was very slow across all libraries exhibiting better
performance under MulticoreBSP and worse under BSPlib. Bitonic sort BTN
had its best performance for $p=16$ or $p=32$ for BSPlib but it was
still twice or thrice slower than PR4. Nothing unexpected with OET.
The three oversampling algorithms were slower than PR4 by roughly 25\%
but were competitive and close to even for smaller problem size 8M and 32M.
GSD and GER had the drawback of multi-way merging which  is currently very 
unoptimized in our implementation. GVR and GER had also to deal
with calling random(). Moreover GVR needed an extra step to avoid
fine grain communication; all keys destined to a single processor
were packed together.

For CONFIG2, an older AMD architecture things changed. GVR primarily
was the leader under OpeMPI, followed by MulticoreBSP and then
BSPlib. The maximum achieved speedup was a lower 9.69. PR2 continued
its slump under OpenMPI and BSPlib but outperformed PR4 in MulticoreBSP.
Under OpenMPI PR4 and GSD and GER had roughly the same performance for
128M with an advantage to PR4, but for 8M and 32M all of GSD and GER
and of course GVR performed better. For MulticoreBSP, PR4 was
better than GSD and GER for the largest size 128M only.
BSPlib extracted better performance from the oversampling methods than
PR4. BTN and OET had miserable performance. The 8 CPUs of the AMD
configuration is probably the reason for that.
Even if this configuration had 32 cores (32 threads) 
vs 20 cores (and 40 threads) of CONFIG1, the maximum speed up
observed and thus efficiency was lower at 9.69 out of 32 vs
16.58 out of 20 for CONFIG1.

For CONFIG3, a more recent architecture the maximum speedup
was only 11.86 for PR4 under OpenMPI. PR4 had the best
overall performance across all libraries. Under OpenMPI
all of PR4, GSD, GVR, GER exhibited maximum performance as well.
The latter three algorithms were only 10\% or less off than PR4
and had better performance for small problem sizes.

CONFIG4 is a universally bad platform across all libraries and 
algorithms.  With 4 cores and 8 hardware supported threads, the 
maximum speed up observed was an 1.83 for 32M.
We can't explain such a result given the performance of the
other three configurations.

\noindent
{\bf Observation 1: Thread size per core.} 
For the AMD platform CONFIG2 one thread per core is a requirement
for extracting best performance. This is in accordance to
remarks by \cite{Langr16} and \cite{Maus11}.
For the Intel configurations CONFIG1 and CONFIG3 the hardware 
thread support utilization was an absolute maximum to extract
maximum performance; in our experiments we used 32 out of a possible 
40 threads on 20 cores.  Experiments for $p=64$ that use more than
two threads per core  were  slow and inefficient.

\noindent
{\bf Observation 2: Hyperthreading.} For the Intel CONFIG1 and CONFIG3 
two threads per core is a requirement for extracting consistently 
better performance thus deviating from \cite{Langr16}. 
Thus the one thread per core recommendation or observation of
\cite{Langr16} might not be current any more for more recent
architectures.

\noindent
{\bf Observation 3: Libraries.} OpenMPI's
library latency makes it perform better in larger problem sizes 
than MulticoreBSP and BSPlib. 
BSPlib with its multiprocessing only support
but low library overhead is extremely competitive
and several times more efficient than OpenMPI, 
despite its age and non support.
It is possible to extract better performance out of it,
if one experiments with is communication parameters and 
their defaults.

\noindent
{\bf Observation 4: 4-round vs 2-round radix-sort.} 
Across the board PR4 is better than PR2.
This is because of the radix $r$ and the  way our 
implementation handles the routing
of radix $r=256$ and $r=65536$ sorting. One may need to
optimize PR2 if it is and can become competitive to PR4.

\noindent
{\bf Observation 5: MBSP modeling SR4 vs PR4.}
We may use equation~\ref{TsG} and equation~\ref{Tp4} to determine
the relative efficiency of a parallel four-round radix-sort.
We have then than
\begin{eqnarray*}
T_s (n,G) / T_p (n/p,G,p) &=& 68nG / \left( 88n/p + 40 \cdot  256 \cdot p \right)G
\end{eqnarray*}
The fraction for large $n$ is approximately $68p/88 \approx 0.77p$.
Thus for $p=4,8,16$ we should not be anticipating speedups higher
than about $3, 6$ and $12$ respectively for CONFIG1.
Indeed this is the case if we consult Table~1 and also Table~3.
It is also the case for Table~2 though the speedup values there
are much less. This might mean that the assumption $g=5G$ might not
be accurate. But it offers a very rough initial estimation of
what one should expect. Note also that the ratio $g/G$ is not
constant; large problem sizes are slow memory bound vs smaller
problem sizes.

\noindent
{\bf Observation 6: MBSP modeling SR4 vs GSD.}
The ratio of equation~\ref{TsG} and the expression for the
running time of GSD as derived from Theorem~\ref{detsorting} which
is $T_s (n/p,G) + 5Gn_{\mathit{max}} \lg{p}$, 
is approximately
\[
T_s (n,G) / (T_s (n/p,G)+6Gn\lg{p}/p) \approx \left( 68 \cdot p \right) / 
\left( 68+6\cdot \lg{(p)} \right) 
\]
This is because for $n< 128,000,000$, $n_{\mathit{max}} \leq 1.2n$
for $\omega_n = \lg{(\lg{n})}$ or so.
For $p=4, 8, 16$ we should
be expecting speedup figures in the range of  $3.4, 6.3,$ and
$11.82$ respectively.
Indeed this is the case as one can deduce from Table~1 - Table~3.
The highest speedup observed for the corresponding
processor/thread sizes is, according to Table~1  in the range 
$2.6-3.16$ for $p=2$, in the range $4.44-5.31$ for $p=4$, and
the range $6.15-8.50$ for $p=16$.
Note that the denominator of the fraction above should also
include asymptotically smaller terms as explained in the
analysis of Theorem~\ref{detsorting}. An asymptotically small
term might be significant for small values of $n$ or $p$.

Note that under Observations 5 and 6 we did not consider the
case for $p=32$.  This is because for say CONFIG1 with 
20 cores, a $p=32$ implies that multiple 
threads (two) are assigned 
to some of the cores. There is a non-homogeneity that any
modelling with MBSP or any other model cannot reliably capture.
Moreover the ratio  $g/G$  is also expected to be non-constant
either.

\noindent
{\bf Observation 7: Smaller problem size sorting.}
Table~5 presents timing results (in microseconds) for small
problem sizes for selected algorithms. The low overhead
BTN fares better than PR4 or SR4 and more often than not
also fares better than GSD. It is hard to beat the simplicity of
programming a BTN for small problem sizes (and values of $p$).

\section{Conclusion}

An experimental study of integer sorting 
on multicores was undertaken using multithreading and
multiprocessing programming approaches resulting in
code that is portable and transportable and  works without
any modifications under three
multithreading or multiprocessing libraries, Open MPI, 
MulticoreBSP, and BSPlib.
We have implemented serial \cite{CLRS} and parallel 
radix-sort for various radixes and also some previously
little explored or unexplored variants of 
bitonic-sort and odd-even transposition sort.
Moreover we implemented two variants of random oversampling
parallel sorting algorithms and made them to work in the context
of integer sorting, as well as a deterministic regular oversampling
parallel sorting algorithm. To the best of our knowledge
this is the first time that so many diverse and varying structure
and concept algorithms have been benchmarked against one another.

We have offered a series of observations obtained through
this evaluation and presented in a systematic way.
Some of those observations have been made previously,
but some of them might not be valid any more for modern
architectures.
Moreover we have expressed the performance of some of
our implementations in the context of the MBSP model 
\cite{G15}. We showed how one can use the model 
to compare the theoretical performance of the 
implementations involving SR4, PR4 and GSD. 
Several conclusions
drawn through this theoretical comparison are in line
with the experimental results we obtained. This would
suggest that MBSP might have merit in studying the
behavior of multicore and multi-memory hierarchy algorithms
and thus be a useful and usable model. We have also highlighted
why precise and accurate performance prediction might be 
difficult to achieve with MBSP or more complex and hierarchical
models. Complex algorithms exhibit complex and difficult to precisely
model memory interactions and patterns. Furthermore, users
interact with memory through third-party libraries that
facilitate parallel multithreading or multiprocessing programming.
Modeling such libraries or their interactions with memory
is next to impossible.


\newpage
\vspace*{13pt}
\centerline{\footnotesize Table~1. Time(sec) for SR4; Speedup for other on CONFIG1}
\noindent
\begin{center}
{\footnotesize
\begin{tabular}{|l r|r|r|r||r|r|r||r|r|r|}\hline
\multicolumn{11}{|c|}{Speedup on Intel Platform}  \\ \hline
\multicolumn{2}{|c|}{}                  & \multicolumn{3}{c||}{OpenMPI}   &
\multicolumn{3}{c||}{MulticoreBSP}       & \multicolumn{3}{c|}{BSPlib} \\ \hline
   &       &  8M & 32M & 128M  &  8M & 32M & 128M           &  8M & 32M & 128M \\ \hline
SR4&$p=1$  &0.250&1.160&5.290  &0.250&1.160&5.290           &0.250&1.160&5.290  \\ \hline
PR4&$p=4 $ &3.47 &3.86 &3.39   &3.24 & 3.67& 3.70           &3.01 &3.29 & 3.12  \\
PR4&$p=8 $ &5.68 &6.82 &7.00   &5.81 & 6.55& 7.18           &5.81 &5.57 & 6.31  \\
PR4&$p=16$ &7.57 &10.84&12.90  &9.61 &10.26&11.65           &10.00&7.68 &11.45  \\
PR4&$p=32$ &4.38 &9.28 &15.88  &11.36&11.95&14.14           &10.00&12.47&16.58  \\ \hline
PR2&$p=4 $ &1.59 &2.22 &2.82   &2.11 & 2.71& 3.12           &0.76 &1.66 & 2.73  \\
PR2&$p=8 $ &1.52 &2.63 &4.30   &3.08 & 4.56& 5.31           &0.62 &1.45 & 3.65  \\
PR2&$p=16$ &1.05 &2.07 &4.65   &5.10 & 5.65& 6.73           &0.56 &0.99 & 3.77  \\
PR2&$p=32$ &0.44 &1.10 &4.59   &5.81 & 6.13& 7.45           &0.46 &0.89 & 3.33  \\ \hline
BTN&$p=4 $ &2.71 &2.76 &2.42   &2.29 & 2.43& 2.58           &2.25 &2.60 & 2.36  \\
BTN&$p=8 $ &3.67 &3.86 &3.67   &2.97 & 3.38& 3.77           &2.87 &3.50 & 3.85  \\
BTN&$p=16$ &4.31 &4.64 &4.70   &3.67 & 4.18& 4.80           &3.62 &3.95 & 4.23  \\
BTN&$p=32$ &3.12 &3.07 &4.05   &3.33 & 3.53& 4.35           &3.33 &3.58 & 4.54  \\ \hline
OET&$p=4 $ &2.57 &2.61 &2.27   &1.89 & 1.99& 2.13           &1.98 &2.20 & 2.06  \\
OET&$p=8 $ &2.90 &3.12 &3.03   &2.19 & 2.43& 2.73           &2.19 &2.54 & 2.80  \\
OET&$p=16$ &2.52 &2.85 &3.37   &2.57 & 2.76& 2.99           &2.42 &2.47 & 2.74  \\
OET&$p=32$ &1.63 &1.59 &1.94   &1.87 & 2.09& 2.44           &1.60 &1.68 & 2.24  \\ \hline
GSD&$p=4 $ &3.16 &3.03 &2.65   &2.84 & 2.87& 2.96           &2.77 &2.96 & 2.66  \\
GSD&$p=8 $ &5.55 &5.37 &4.62   &4.71 & 5.00& 5.20           &4.31 &3.42 & 5.06  \\
GSD&$p=16$ &8.62 &8.85 &9.18   &7.35 & 7.83& 9.31           &8.33 &7.53 & 7.97  \\
GSD&$p=32$ &10.0 &9.13 &11.23  &9.61 & 7.29&10.70           &10.41&10.84&12.53  \\ \hline
GVR&$p=4 $ &2.90 &2.78 &3.24   &2.87 & 2.78& 3.10           &2.87 &2.75 & 3.20  \\
GVR&$p=8 $ &5.00 &5.13 &5.63   &4.80 & 5.00& 5.76           &4.31 &4.42 & 5.54  \\
GVR&$p=16$ &8.06 &8.46 &9.65   &7.81 & 8.28& 9.32           &7.81 &8.22 & 9.34  \\
GVR&$p=32$ &8.62 &9.66 &11.91  &8.62 & 9.13&11.52           &9.61 &10.54&12.13  \\ \hline
GER&$p=4 $ &3.28 &3.23 &2.81   &2.94 & 3.01& 3.11           &2.74 &3.05 & 2.76  \\
GER&$p=8 $ &5.00 &5.34 &4.89   &4.90 & 4.91& 5.55           &4.16 &4.01 & 5.35  \\
GER&$p=16$ &8.62 &8.72 &8.74   &7.57 & 7.58& 9.21           &8.33 &7.68 & 7.82  \\
GER&$p=32$ &9.61 &9.35 &11.96  &10.86& 7.43& 9.94           &9.61 &10.94&12.83  \\ \hline
\end{tabular}}
\end{center}

\newpage
\vspace*{13pt}
\centerline{\footnotesize Table~2. Time(sec) for SR4; Speedup for other on CONFIG2}
\noindent
\begin{center}
{\footnotesize
\begin{tabular}{|l r|r|r|r||r|r|r||r|r|r|}\hline
\multicolumn{11}{|c|}{Speedup on Intel Platform}  \\ \hline
\multicolumn{2}{|c|}{}                  & \multicolumn{3}{c||}{OpenMPI}   &
\multicolumn{3}{c||}{MulticoreBSP}       & \multicolumn{3}{c|}{BSPlib} \\ \hline
   &       &  8M & 32M & 128M  &  8M & 32M & 128M           &  8M & 32M & 128M \\ \hline
SR4&$p=1$  &0.340&1.360&7.300  &0.340&1.360&7.300           &0.340&1.360&7.300  \\ \hline
PR4&$p=4 $ &2.31 &2.46 &3.25   &2.32 &2.29 &3.10            &1.79 &1.84 &2.45   \\
PR4&$p=8 $ &3.90 &4.01 &5.32   &3.73 &3.56 &4.92            &3.11 &2.99 &3.91   \\
PR4&$p=16$ &4.41 &5.48 &7.56   &4.59 &4.75 &6.59            &4.72 &4.30 &5.82   \\
PR4&$p=32$ &4.53 &5.25 &7.79   &6.66 &5.35 &7.65            &4.78 &4.57 &6.50   \\ \hline
PR2&$p=4 $ &1.39 &1.78 &2.53   &2.37 &2.93 &4.36            &0.67 &1.75 &3.61   \\
PR2&$p=8 $ &1.21 &1.85 &2.81   &3.57 &4.40 &6.57            &0.47 &1.56 &4.49   \\
PR2&$p=16$ &0.85 &1.60 &2.84   &4.30 &4.54 &6.41            &0.29 &1.04 &4.01   \\
PR2&$p=32$ &0.54 &1.10 &2.27   &4.47 &5.31 &7.94            &0.18 &0.64 &2.86   \\ \hline
BTN&$p=4 $ &1.89 &1.85 &2.46   &2.09 &1.97 &2.69            &1.87 &1.85 &2.36   \\
BTN&$p=8 $ &2.44 &2.21 &2.84   &2.22 &2.24 &2.94            &2.31 &2.09 &2.68   \\
BTN&$p=16$ &2.06 &1.73 &2.27   &1.78 &1.87 &2.47            &1.97 &1.67 &2.20   \\
BTN&$p=32$ &1.49 &1.20 &1.56   &1.45 &1.19 &1.60            &1.45 &1.16 &1.52   \\ \hline
OET&$p=4 $ &1.67 &1.63 &2.20   &1.71 &1.68 &2.23            &1.38 &1.38 &1.76   \\
OET&$p=8 $ &1.98 &1.82 &2.39   &2.03 &1.97 &2.62            &1.57 &1.47 &1.89   \\
OET&$p=16$ &1.25 &1.08 &1.42   &1.25 &1.36 &1.79            &1.23 &1.04 &1.37   \\
OET&$p=32$ &0.63 &0.53 &0.69   &0.63 &0.59 &0.81            &0.66 &0.54 &0.72   \\ \hline
GSD&$p=4 $ &2.65 &2.47 &3.16   &2.51 &2.40 &3.14            &2.37 &2.28 &2.87   \\
GSD&$p=8 $ &5.31 &4.44 &5.16   &4.30 &3.86 &4.79            &4.59 &4.04 &4.67   \\
GSD&$p=16$ &8.50 &6.15 &7.23   &5.48 &4.90 &6.30            &7.39 &5.69 &6.72   \\
GSD&$p=32$ &12.14&6.80 &6.93   &7.90 &4.54 &6.38            &10.62&6.26 &6.64   \\ \hline
GVR&$p=4 $ &2.26 &2.15 &3.03   &2.41 &2.14 &2.97            &2.34 &2.29 &2.76   \\
GVR&$p=8 $ &4.59 &4.04 &5.40   &3.95 &3.76 &4.97            &4.30 &3.89 &4.77   \\
GVR&$p=16$ &7.39 &6.21 &7.92   &5.66 &5.59 &7.08            &6.80 &5.86 &7.65   \\
GVR&$p=32$ &9.71 &7.39 &9.69   &9.18 &6.41 &9.31            &8.94 &7.01 &8.93   \\ \hline
GER&$p=4 $ &2.74 &2.63 &3.36   &2.59 &2.46 &3.31            &2.55 &2.41 &3.04   \\
GER&$p=8 $ &4.85 &4.31 &5.50   &4.35 &3.82 &4.95            &4.35 &4.04 &4.93   \\
GER&$p=16$ &8.50 &6.15 &7.13   &5.76 &4.78 &6.29            &7.55 &5.69 &6.66   \\
GER&$p=32$ &12.59&6.86 &7.22   &11.72&4.72 &6.31            &10.62&6.38 &6.73   \\ \hline
\end{tabular}}
\end{center}

\newpage
\vspace*{13pt}
\centerline{\footnotesize Table~3. Time(sec) for SR4; Speedup for other on CONFIG3}
\noindent
\begin{center}
{\footnotesize
\begin{tabular}{|l r|r|r|r||r|r|r||r|r|r|}\hline
\multicolumn{11}{|c|}{Speedup on Intel Platform}  \\ \hline
\multicolumn{2}{|c|}{}                  & \multicolumn{3}{c||}{OpenMPI}   &
\multicolumn{3}{c||}{MulticoreBSP}       & \multicolumn{3}{c|}{BSPlib} \\ \hline
   &       &  8M & 32M & 128M  &  8M & 32M & 128M           &  8M & 32M & 128M \\ \hline
SR4&$p=1$  &0.241&1.056&4.200  &0.241&1.056&4.200           &0.241&1.056&4.200  \\ \hline
PR4&$p=4 $ &3.21 &3.41 &3.36   &3.05 &3.25 &3.20            &2.53 &2.54 &2.63   \\
PR4&$p=8 $ &5.23 &6.10 &6.14   &5.47 &5.67 &4.46            &4.08 &4.73 &4.62   \\
PR4&$p=16$ &7.30 &9.96 &10.16  &8.92 &8.87 &8.95            &7.53 &7.23 &7.54   \\
PR4&$p=32$ &4.54 &10.66&11.86  &11.47&11.00&11.32           &8.92 &10.77&11.32  \\ \hline
PR2&$p=4 $ &1.46 &1.87 &1.91   &1.94 &2.30 &2.32            &0.54 &1.35 &2.00   \\
PR2&$p=8 $ &1.44 &2.35 &2.65   &3.08 &3.89 &3.80            &0.41 &1.34 &2.66   \\
PR2&$p=16$ &1.00 &1.83 &2.88   &4.46 &5.00 &5.10            &0.31 &1.06 &2.62   \\
PR2&$p=32$ &0.42 &1.07 &2.03   &5.60 &5.70 &6.00            &0.18 &0.74 &1.86   \\ \hline
BTN&$p=4 $ &2.73 &2.61 &2.53   &2.80 &2.62 &2.59            &2.29 &2.31 &2.22   \\
BTN&$p=8 $ &3.76 &3.71 &3.62   &3.70 &3.71 &3.16            &2.90 &3.15 &2.92   \\
BTN&$p=16$ &4.30 &4.31 &4.28   &4.08 &4.24 &4.19            &3.25 &3.79 &3.58   \\
BTN&$p=32$ &2.43 &3.79 &3.68   &3.59 &3.56 &3.84            &3.12 &3.23 &3.08   \\ \hline
OET&$p=4 $ &2.61 &2.46 &2.40   &2.48 &2.39 &2.33            &1.83 &1.91 &1.85   \\
OET&$p=8 $ &2.90 &2.97 &2.93   &3.05 &3.15 &2.96            &2.21 &2.34 &2.20   \\
OET&$p=16$ &2.36 &2.50 &2.56   &3.12 &3.36 &3.12            &2.36 &2.57 &2.41   \\
OET&$p=32$ &1.56 &2.08 &1.84   &2.13 &2.26 &2.36            &1.52 &1.48 &1.38   \\ \hline
GSD&$p=4 $ &3.12 &2.84 &2.73   &2.93 &2.65 &2.69            &2.59 &2.56 &2.44   \\
GSD&$p=8 $ &5.35 &5.00 &4.45   &4.91 &4.61 &4.00            &4.15 &4.45 &3.77   \\
GSD&$p=16$ &8.31 &8.00 &8.04   &6.88 &6.72 &6.38            &5.60 &7.23 &7.79   \\
GSD&$p=32$ &9.64 &10.25&10.16  &9.26 &5.61 &7.48            &9.26 &10.15&7.25   \\ \hline
GVR&$p=4 $ &2.90 &2.65 &2.92   &3.17 &2.77 &2.87            &2.56 &2.35 &2.70   \\
GVR&$p=8 $ &5.02 &4.82 &4.74   &5.12 &4.84 &4.52            &4.08 &4.47 &4.42   \\
GVR&$p=16$ &7.53 &7.59 &7.52   &7.53 &7.23 &6.22            &5.23 &7.13 &7.17   \\
GVR&$p=32$ &8.31 &9.26 &9.48   &8.60 &8.38 &8.30            &8.92 &9.42 &9.52   \\ \hline
GER&$p=4 $ &3.21 &2.97 &2.87   &3.08 &2.89 &2.82            &2.70 &2.65 &2.58   \\
GER&$p=8 $ &4.82 &4.93 &4.74   &5.02 &4.63 &4.32            &3.88 &4.41 &3.92   \\
GER&$p=16$ &8.31 &7.82 &7.85   &7.30 &6.17 &6.26            &6.34 &7.23 &7.48   \\
GER&$p=32$ &9.64 &10.15&10.37  &10.04&5.67 &7.16            &9.26 &10.15&9.97   \\ \hline
\end{tabular}}
\end{center}

\newpage
\vspace*{13pt}
\centerline{\footnotesize Table~4. Time(sec) for SR4; Speedup for other on CONFIG4}
\noindent
\begin{center}
{\footnotesize
\begin{tabular}{|l r|r|r|r||r|r|r||r|r|r|}\hline
\multicolumn{11}{|c|}{Speedup on Intel Platform}  \\ \hline
\multicolumn{2}{|c|}{}                  & \multicolumn{3}{c||}{OpenMPI}   &
\multicolumn{3}{c||}{MulticoreBSP}       & \multicolumn{3}{c|}{BSPlib} \\ \hline
   &       &  8M & 32M & 128M  &  8M & 32M & 128M           &  8M & 32M & 128M \\ \hline
SR4&$p=1$  &0.054&0.224&0.840  &0.054&0.224&0.840           &0.054&0.224&0.840  \\ \hline
PR4&$p=2 $ &1.35 &1.34 &1.30   &1.01 &1.01 &1.03            &1.28 &1.28 &1.23   \\
PR4&$p=4 $ &1.38 &1.86 &1.33   &1.22 &1.20 &1.15            &1.63 &1.70 &1.54   \\
PR4&$p=8 $ &1.68 &1.83 &1.67   &1.20 &1.22 &1.10            &1.63 &1.72 &1.56   \\ \hline
PR2&$p=2 $ &0.69 &0.80 &0.78   &0.80 &0.82 &0.82            &0.37 &0.67 &0.79   \\
PR2&$p=4 $ &0.53 &0.79 &0.67   &0.80 &0.85 &0.81            &0.30 &0.64 &0.82   \\
PR2&$p=8 $ &0.34 &0.50 &0.50   &0.56 &0.62 &0.59            &0.18 &0.43 &0.57   \\ \hline
BTN&$p=2 $ &1.08 &1.10 &1.03   &1.03 &1.04 &0.97            &1.08 &1.044&0.95   \\
BTN&$p=4 $ &0.85 &1.05 &0.79   &0.93 &0.93 &0.95            &0.98 &1.00 &0.91   \\
BTN&$p=8 $ &0.70 &0.72 &0.66   &0.60 &0.62 &0.56            &0.68 &0.70 &0.63   \\ \hline
OET&$p=2 $ &1.10 &1.12 &1.05   &1.08 &1.04 &0.98            &1.14 &1.10 &1.03   \\
OET&$p=4 $ &0.83 &1.05 &0.80   &0.96 &0.95 &0.87            &1.01 &1.03 &0.93   \\
OET&$p=8 $ &0.65 &0.70 &0.63   &0.59 &0.62 &0.56            &0.65 &0.67 &0.61   \\ \hline
GSD&$p=2 $ &1.14 &1.15 &1.08   &1.28 &1.24 &1.12            &1.35 &1.29 &1.22   \\
GSD&$p=4 $ &1.12 &1.35 &0.97   &1.54 &1.40 &1.33            &1.58 &1.58 &1.42   \\
GSD&$p=8 $ &1.31 &1.33 &1.16   &1.28 &1.37 &1.17            &1.42 &1.43 &1.22   \\ \hline
GVR&$p=2 $ &0.96 &0.97 &0.92   &1.14 &0.99 &0.97            &1.22 &1.15 &1.10   \\
GVR&$p=4 $ &0.91 &1.19 &0.86   &1.58 &1.53 &1.35            &1.54 &1.58 &1.40   \\
GVR&$p=8 $ &1.14 &1.19 &1.06   &1.42 &1.48 &1.27            &1.50 &1.50 &1.32   \\ \hline
GER&$p=2 $ &1.20 &1.28 &1.12   &1.31 &1.15 &1.14            &1.38 &1.31 &1.23   \\
GER&$p=4 $ &1.08 &1.43 &1.04   &1.54 &1.56 &1.35            &1.63 &1.63 &1.42   \\
GER&$p=8 $ &1.31 &1.36 &1.21   &1.28 &1.38 &1.23            &1.38 &1.42 &1.28   \\ \hline
\end{tabular}}
\end{center}

\newpage
\vspace*{13pt}
\centerline{\footnotesize Table~5. Small problem size experiments on CONFIG3}
\noindent
\begin{center}
{\footnotesize
\begin{tabular}{|l r|r|r|r||r|r|r||r|r|r|}\hline
\multicolumn{11}{|c|}{Time in microseconds}  \\ \hline
\multicolumn{2}{|c|}{}                  & \multicolumn{3}{c||}{OpenMPI}   &
\multicolumn{3}{c||}{MulticoreBSP}       & \multicolumn{3}{c|}{BSPlib} \\ \hline
   &       &8192 &32768&131072           &8192 &32768&131072           &8192 &32768&131072\\ \hline
SR4&$p=1$  & 109 & 483 & 2076            & 109 & 483 & 2076            & 109 & 483 &2076   \\ \hline
PR4&$p=4 $ & 678 & 855 & 1400            & 282 & 397 &  886            & 829 & 209 &2571   \\
PR4&$p=16$ &1404 &1218 & 1933            & 633 & 655 &  987            & 417 & 482 &4356   \\
PR4&$p=32$ &3338 &8225 & 6818            &1259 &1307 & 1512            & 516 &6807 &6067   \\  \hline
BTN&$p=4 $ & 102 & 268 &  896            &  72 & 274 & 1105            &  86 & 298 & 944   \\
BTN&$p=16$ & 254 & 362 & 1062            & 151 & 236 &  699            & 176 & 304 & 731   \\
BTN&$p=32$ & 972 &1990 & 2608            & 290 & 382 &  816            & 401 & 426 & 904   \\ \hline
GSD&$p=4 $ & 187 & 338 &  855            &  93 & 249 &  908            & 117 & 261 & 744   \\
GSD&$p=16$ & 406 & 408 &  909            & 181 & 237 &  484            & 416 & 495 & 645   \\
GSD&$p=32$ &2030 &2347 & 2589            & 348 & 374 &  619            &1120 &1029 &1184   \\ \hline
GVR&$p=4 $ & 237 & 424 & 1055            & 207 & 425 & 1248            & 135 & 339 & 938   \\
GVR&$p=16$ & 541 & 534 & 1075            & 496 & 665 & 1044            & 373 & 463 & 672   \\
GVR&$p=32$ &1248 &2376 & 2620            & 988 &1252 & 1762            & 987 & 909 &1071   \\ \hline
GER&$p=4 $ & 211 & 373 &  862            & 138 & 293 &  965            & 116 & 286 & 768   \\
GER&$p=16$ & 492 & 498 & 1006            & 464 & 594 &  890            & 411 & 531 & 654   \\
GER&$p=32$ &1207 &2455 & 2586            &1017 &1216 & 1646            &1080 &1030 &1172   \\ \hline
\end{tabular}}
\end{center}

\end{document}